\documentclass[nofootinbib,reprint,amssymb,superscriptaddress,nofootinbib]{revtex4-2}
\usepackage{subfig}
\usepackage{amsmath,amsfonts,amssymb,mathtools}
\usepackage{tasks}
\usepackage{float}
\usepackage{csquotes}
\usepackage{changepage}
\usepackage{float}
\usepackage{bbm,bm}
\usepackage[normalem]{ulem}
\usepackage{comment}
\usepackage{physics}
\usepackage{xcolor}
\definecolor{green}{rgb}{0.1,0.1,0.1}
\usepackage[utf8]{inputenc}
\usepackage{pgfplots}
\definecolor{blueprl}{RGB}{46,48,146}
\usepgfplotslibrary{groupplots}

\usepackage{lipsum}
\newtheorem{theorem}{Theorem}

\newenvironment{proof}[1][Proof]{\noindent\textbf{#1.} }{\ \rule{0.5em}{0.5em}}

\usepackage{tikz}
\usepackage{lipsum}
\usetikzlibrary{calc,shadings,patterns}
\makeatletter
\tikzset{%
  remember picture with id/.style={%
    remember picture,
    overlay,
    save picture id=#1,
  },
  save picture id/.code={%
    \edef\pgf@temp{#1}%
    \immediate\write\pgfutil@auxout{%
      \noexpand\savepointas{\pgf@temp}{\pgfpictureid}}%
  },
  if picture id/.code args={#1#2#3}{%
    \@ifundefined{save@pt@#1}{%
      \pgfkeysalso{#3}%
    }{
      \pgfkeysalso{#2}%
    }
  }
}

\def\savepointas#1#2{%
  \expandafter\gdef\csname save@pt@#1\endcsname{#2}%
}

\def\tmk@labeldef#1,#2\@nil{%
  \def\tmk@label{#1}%
  \def\tmk@def{#2}%
}

\tikzdeclarecoordinatesystem{pic}{%
  \pgfutil@in@,{#1}%
  \ifpgfutil@in@%
    \tmk@labeldef#1\@nil
  \else
    \tmk@labeldef#1,(0pt,0pt)\@nil
  \fi
  \@ifundefined{save@pt@\tmk@label}{%
    \tikz@scan@one@point\pgfutil@firstofone\tmk@def
  }{%
  \pgfsys@getposition{\csname save@pt@\tmk@label\endcsname}\save@orig@pic%
  \pgfsys@getposition{\pgfpictureid}\save@this@pic%
  \pgf@process{\pgfpointorigin\save@this@pic}%
  \pgf@xa=\pgf@x
  \pgf@ya=\pgf@y
  \pgf@process{\pgfpointorigin\save@orig@pic}%
  \advance\pgf@x by -\pgf@xa
  \advance\pgf@y by -\pgf@ya
  }%
}

\makeatother

\def\ANU{Centre for Quantum Computation and Communication Technology, Department of Quantum Science, Australian National University, Canberra, ACT 2601, Australia.}

    \def\Astar{Institute of Materials Research and Engineering, Agency for Science Technology and Research (A*STAR), 2 Fusionopolis Way, 08-03 Innovis 138634, Singapore}

\def\Astartwo{Institute of High Performance Computing, Agency for Science, Technology and Research (A*STAR), 1 Fusionopolis Way, \#16-16 Connexis, Singapore 138632, Republic of Singapore}

\def\strathclyde{SUPA Department of Physics, The University of Strathclyde, Glasgow, G4 0NG, UK}
\def\Harvard{Department of Physics, Harvard University, Cambridge, Massachusetts 02138, USA}

\begin{document}
\title{Attainability of quantum state discrimination bounds with collective measurements on finite copies}
\author{Lorc{\'a}n O. Conlon}
\email{lorcanconlon@gmail.com}
\affiliation{\Astar}
\author{Jin Ming Koh}
\affiliation{\Astartwo}
\affiliation{\Harvard}
\author{Biveen Shajilal}
\affiliation{\Astar}
\author{Jasminder Sidhu}
\affiliation{\strathclyde}
\author{Ping Koy Lam}
\affiliation{\Astar}
\affiliation{\ANU}
\author{Syed M. Assad}
\email{cqtsma@gmail.com}
\affiliation{\Astar}
\affiliation{\ANU}

\begin{abstract}
One of the fundamental tenets of quantum mechanics is that non-orthogonal states cannot be distinguished perfectly. When distinguishing multiple copies of a mixed quantum state, a collective measurement, which generates entanglement between the different copies of the unknown state, can achieve a lower error probability than non-entangling measurements. The error probability that can be attained using a collective measurement on a finite number of copies of the unknown state is given by the Helstrom bound. In the limit where we can perform a collective measurement on asymptotically many copies of the quantum state, the quantum Chernoff bound gives the attainable error probability. It is natural to ask then what strategies can be employed to reach these two bounds -- is entanglement across all available modes always necessary or can other, experimentally more simple, measurements saturate these bounds. In this work we address this question. We find analytic expressions for the Helstrom bound for arbitrarily many copies of the unknown state in simple qubit examples. Using these analytic expressions, we investigate whether the quantum Chernoff bound can be saturated by repeatedly implementing the $M$-copy Helstrom measurement. We also investigate the necessary conditions to saturate the $M$-copy Helstrom bound. It is known that a collective measurement on all $M$-copies of the unknown state is always sufficient to saturate the $M$-copy Helstrom bound. However, general conditions for when such a measurement is necessary to saturate the Helstrom bound remain unknown. We investigate specific measurement strategies which involve entangling operations on fewer than all $M$-copies of the unknown state. For many regimes we find that a collective measurement on all $M$-copies of the unknown state is necessary to saturate the $M$-copy Helstrom bound. 
\end{abstract}
\maketitle

\section{Introduction}
Quantum state discrimination can be described as the task of discriminating between $M$ copies of two unknown quantum states $\rho_+^{\otimes M}$ or $\rho_-^{\otimes M}$. Due to the no-cloning theorem, it is impossible to perform this task perfectly unless the states to be discriminated are orthogonal to one another~\cite{wootters1982single,dieks1982communication}. Therefore, it has become an important task to find measurements which minimise the average probability of making an error while remaining experimentally feasible. When there is only a single copy of the unknown quantum state, i.e. $M=1$, the optimal measurement is simple to implement, not requiring any entangling operations. Along this line, there have been several experimental studies focused on optimally distinguishing single copies of pure states~\cite{cook2007optical,wittmann2008demonstration,bartuuvskova2008programmable,waldherr2012distinguishing,becerra2013implementation,izumi2020experimental,izumi2021adaptive,sidhu2021quantum,gomez2022experimental}. For $M>1$, there are two main classes of measurements: (1) adaptive measurements which rely only on local operations and classical communication (LOCC) between the $M$ different modes of the unknown quantum states and (2) collective measurements which generate entanglement between all $M$ copies of the unknown quantum state. Clearly the latter is a more powerful class of measurement which includes the former as a special case. However, the cost of this is that collective measurements are more experimentally complex to implement. As a result of this, adaptive measurements, or LOCC measurements, have been well studied and shown to outperform non-adaptive measurements~\cite{walgate2000local,becerra2013experimental,higgins2009mixed}. In contrast, the advantage of collective measurements over LOCC measurements in this context has only been demonstrated experimentally recently~\cite{conlon2023discriminating,zhou2023experimental}. Note that the advantage of collective measurements has been experimentally demonstrated in several other areas of quantum information~\cite{roccia2017entangling,hou2018deterministic,conlon2023approaching,parniak2018beating,wu2019experimentally,yuan2020direct,wu2020minimizing}. Throughout this work we shall refer to collective measurements as those acting on multiple copies of the same state\footnote{However, it is important to note that many other tasks in quantum information also benefit from measurements in an entangled basis, which do not necessarily act on multiple copies of the same quantum state, such as quantum metrology~\cite{marciniak2022optimal}, communications~\cite{delaney2022demonstration,crossman2023quantum} and orienteering~\cite{gisin1999spin,jeffrey2006optical,tang2020experimental}.}.

Theoretically, the role of collective measurements in state discrimination is reasonably well understood. Helstrom found a simple expression for the minimal error probability that can be attained when performing an entangling measurement on all $M$-copies of the unknown state. This has become known as the Helstrom bound~\cite{helstrom1969quantum}. Helstrom's work also provides the entangling measurement which can saturate his bound. In certain situations however, collective measurements are not needed to saturate the Helstrom bound. Trivially, when $M=1$, there is only one copy of the unknown state available and so collective measurements are not possible. When discriminating two pure states it is known that collective measurements offer no advantage over separable measurements~\cite{brody1996minimum,ban1997accessible,acin2005multiple}. (Note however, that in slightly different settings collective measurements may outperform separable measurements for discriminating certain families of pure states~\cite{peres1991optimal,massar1995optimal,chitambar2013revisiting}.) Additionally, when the two states to be discriminated are simultaneously diagonalisable, the Helstrom bound can be saturated through the maximum likelihood measurement which requires only local operations, see Ref.~\cite{audenaert2014upper} and appendix~\ref{apen:simdiag}. However, in general for $M>1$ when discriminating between two mixed states $\rho_+^{\otimes M}$ or $\rho_-^{\otimes M}$, collective measurements may be required to saturate the Helstrom bound. In spite of the importance of this question, concrete results regarding when collective measurements are required to saturate the Helstrom bound remain largely unknown (Necessary and sufficient conditions that the optimal measurement must satisfy are known, see e.g. Refs.~\cite{helstrom1969quantum,ban1997optimum} and appendix~\ref{apen:POVMopt}). Much of our knowledge in this area comes from numerical results, where it has been shown that for discriminating qubit states collective measurements outperform separable measurements~\cite{calsamiglia2010local,higgins2011multiple,conlon2023discriminating}.

When we take the limit of performing collective measurements on a large number of copies of the unknown state, the error probability is given by the quantum Chernoff bound~\cite{audenaert2007discriminating,Nussbaum2009chernoff}. This is the quantum analogue of the Chernoff bound in classical statistics~\cite{chernoff1952measure}. Interestingly, the conditions for what measurement strategy is optimal changes in the asymptotic limit. For example, it is known that a fixed single-copy measurement is capable of asymptotically approaching this minimum error probability when discriminating two states, one of which is pure~\cite{acin2005multiple,calsamiglia2008quantum}. Additionally, we note that Hayashi has shown that LOCC measurements do not outperform fixed local measurements in terms of the asymptotic error rate~\cite{hayashi2009discrimination}. Tables~\ref{T1} and \ref{T2} summarise these known results regarding when different measurement strategies and bounds are equal. We compare both the asymptotic (Table~\ref{T1}) and  finite-copy (Table~\ref{T2}) regime. These tables also indicate where we see our work as fitting into the existing literature.

In this paper we investigate the attainability of both the quantum Chernoff bound and the Helstrom bound given finite resources. For the attainability of the quantum Chernoff bound we investigate the rate at which the error probability specified by the quantum Chernoff bound is approached with collective measurements on a finite number of copies of the unknown state. We do this by considering repeated implementations of the Helstrom measurement on groups of quantum states. This investigates the scenario where entangling measurements can be repeatedly implemented on small subsets of the available quantum states.

 For investigating the attainability of the Helstrom bound we ask whether collective measurements on all $M$ copies of a quantum state are necessary to saturate the Helstrom bound. We address this question using parameterised quantum circuits optimised to find measurements which minimise the error probablity. Although there exists much numerical evidence that there is a gap between the Helstrom bound and the LOCC bound Ref.~\cite{calsamiglia2010local}, to the best of our knowledge whether the Helstrom bound can be saturated by intermediate measurements remains unknown and relatively unexplored. By intermediate measurements we mean measurements capable of generating entanglement between more than 1, but not all $M$ copies of the unknown quantum state. We investigate how close these intermediate measurements can be to the Helstrom bound. Our investigations also reveal the optimal strategies given limited entangling resources. For example, given $M$ copies of the unknown state, and a device capable of faithfully performing an entangling measurement on $90\%$ of these copies, is it better to perform this entangling operation on the first $90\%$ or the last $90\%$? Our results indicate that the later option is superior. We then extend these results from LOCC measurements to positive partial transpose (PPT) measurements which are a stronger class of measurements than LOCC, but still weaker than entangling measurements on all $M$ copies. Interestingly in this setting we find different results depending on whether we consider discriminating between two or more quantum states. 

We note that there have been several previous studies which have examined state discrimination in the finite copy setting. Refs.~\cite{pereira2023analytical,audenaert2012quantum,rouze2017finite,li2014second} provides bounds on the error probability attainable in the $M$-copy asymmetric state discrimination setting. As we provide exact expressions for the $M$-copy Helstrom bound, our work can be thought of as building on these previous more general works, by providing exact solutions for several simple examples. We also note that there are several works bounding other closely related quantities, including the channel capacity and quantum channel discrimination, in the finite-copy regime~\cite{tomamichel2013hierarchy,wang2012one,tomamichel2016quantum,bergh2024parallelization}.

%

The rest of the paper is set out as follows. In section~\ref{secprelim}, we introduce the preliminary material: the Helstrom bound, the quantum Chernoff bound and the questions we wish to address. In section~\ref{secres1}, we present our main results on the attainability of the quantum Chernoff bound. We present analytic expressions for the Helstrom bound when discriminating multiple copies of qubit states. We use this to investigate the attainability of the Chernoff bound when repeatedly implementing the $M$-copy Helstrom measurement. In section~\ref{secres2} we investigate the requirements to saturate the $M$-copy Helstrom bound, building on the results in the previous section. Finally, in section~\ref{secconc}, we present our final remarks and conclude.


\begin{table*}[t]
\caption{\label{T1}\textbf{Conditions for equality between different measurement strategies and bounds in quantum state discrimination in the asymptotic limit.} Each box gives the known results about the conditions for equality between the corresponding measurement strategies or bounds in that row and column.  $\kappa$, $P^{\text{H}}_\text{e}(M)$, $P_\text{e,LOCC}$ and $P_\text{e,LO}$ denote the quantum Chernoff bound, the $M$-copy Helstrom bound, the LOCC error probability and the error probability when performing the same local operation respectively. In this table we only consider the large $M$ limit.}
  \centering
\begin{tabular}[t]{ | p{2cm} |  p{4.2cm}| p{4.2cm}| p{4.2cm} |} 
  \hline
 &  $P^{\text{H}}_\text{e}(M)$& $P_\text{e,LOCC}$& $P_\text{e,LO}$ \\ 
 \hline
 $\kappa$ & Equal by definition~\cite{audenaert2007discriminating,Nussbaum2009chernoff}. 
 &Equal for pure states~\cite{brody1996minimum,ban1997accessible,acin2005multiple}.&Equal up to a constant when one of the states is pure~\cite{acin2005multiple,calsamiglia2008quantum}.
 \\ \hline
$P^{\text{H}}_\text{e}(M)$   & \hspace{1.6cm}\scalebox{4}{\raisebox{-2.7ex}{\checkmark}} &Equal for pure states~\cite{brody1996minimum,ban1997accessible,acin2005multiple}.\newline\noindent Gap between the two conjectured to persist in the asymptotic limit for mixed qubit states in Ref.~\cite{calsamiglia2010local}.&By virtue of the previous column, conjectured to be unequal in the asymptotic limit for mixed qubit states~\cite{calsamiglia2010local}. \newline\noindent By virtue of the above row, equal up to a constant when one state is pure~\cite{acin2005multiple,calsamiglia2008quantum}. \\ 
   \hline
  $P_\text{e,LOCC}$&\hspace{1.6cm}\scalebox{4}{\raisebox{-3ex}{\checkmark}} &\hspace{1.6cm}\scalebox{4}{\raisebox{-3ex}{\checkmark}}  & By virtue of the above row, equal up to a constant when one state is pure~\cite{acin2005multiple,calsamiglia2008quantum}.\newline\noindent Not equal for mixed states~\cite{higgins2011multiple}. \newline\noindent Fixed measurements can achieve the same scaling as adaptive~\cite{hayashi2009discrimination}.\\ \hline 
\end{tabular}
\end{table*}


\begin{table*}[t]
\caption{\label{T2}\textbf{Conditions for equality between different measurement strategies in quantum state discrimination in the finite copy limit.} Each box gives the known results about the conditions for equality between the corresponding measurement strategies or bounds in that row and column.  $\kappa$, $P^{\text{H}}_\text{e}(M)$, $P_\text{e,LOCC}$ and $P_\text{e,LO}$ denote the $M$-copy quantum Chernoff bound, the $M$-copy Helstrom bound, the LOCC error probability and the error probability when performing the same local operation respectively. In this table we consider any finite positive integer value of $M$. Note that we include the quantum Chernoff bound although it is, by definition, an asymptotic limit. For this row, we refer to repeated implementations of the corresponding optimal $M$-copy measurement, and whether that can asymptotically saturate the quantum Chernoff bound.}
  \centering
\begin{tabular}[t]{ | p{2cm} |  p{4.2cm} | p{4.2cm} | p{4.2cm} |} 
  \hline
 &  $P^{\text{H}}_\text{e}(M)$& $P_\text{e,LOCC}$& $P_\text{e,LO}$ \\ 
 \hline
 $\kappa$ & Investigated in this work 
 & Indirectly investigated in this work&Indirectly investigated in this work \\ 
   \hline
$P^{\text{H}}_\text{e}(M)$   &\hspace{1.6cm}\scalebox{4}{\raisebox{-2.7ex}{\checkmark}} &\noindent For qubit states can be equal or unequal depending on the number of copies and the prior probability~\cite{conlon2023discriminating,calsamiglia2010local}. \newline\noindent Equal for pure states~\cite{brody1996minimum,ban1997accessible,acin2005multiple}. \newline\noindent This is further investigated in this work.&In general not equal~\cite{higgins2011multiple,calsamiglia2010local}. \newline\noindent Equal for any two simultaneously diagonalisable states through the maximum likelihood measurement, see Ref.~\cite{audenaert2014upper} and appendix~\ref{apen:simdiag}. 
  \\ 
   \hline
  $P_\text{e,LOCC}$& \hspace{1.8cm}\scalebox{2.5}{\raisebox{-1.1ex}{\checkmark}}&\hspace{1.8cm}\scalebox{2.5}{\raisebox{-1ex}{\checkmark}} & In general not equal, see for example Refs~\cite{higgins2011multiple,calsamiglia2010local}\\ \hline 
\end{tabular}
\end{table*}

\section{Preliminaries}
\label{secprelim}
 Given $M$ copies of the unknown state, $\rho_{\pm}^{\otimes M}$, our task is to decide which state we are given with the minimum probability of error. Denoting the prior probability for the unknown state to be the state $\rho_{+}^{\otimes M}$ as $q$, the error probability is given by
\begin{equation}
\label{eq:errorprobabilitybasic}
P_\text{e}=qP_{-|+}+(1-q)P_{+|-}\;,
\end{equation}
where $P_{-|+}$ is the probability of guessing the state $\rho_{-}^{\otimes M}$ when given the state $\rho_{+}^{\otimes M}$ and $P_{+|-}$ is similarly defined. 

\subsection{Helstrom bound}
In general, the optimal measurement to distinguish $\rho_{\pm}^{\otimes M}$, will be a collective measurement which involves entangling operations between all $M$ modes. The minimum error probability, in this case, is given by the Helstrom bound~\cite{helstrom1969quantum} 
\begin{equation}
\label{eq:helstromdefinition}
P^{\text{H}}_\text{e}(M)=\frac{1}{2}\bigg(1-\left\lVert q\rho_+^{\otimes M}-(1-q)\rho_-^{\otimes M}\right\rVert\bigg)\;,
\end{equation}
where $P^{\text{H}}_\text{e}(M)$ denotes the error probability which can be attained by collective measurements and $\lVert A\rVert$ is the sum of the absolute values of the eigenvalues of $A$. It is known that the Helstrom bound can be saturated by measuring the following observable, which may or may not correspond to a collective measurement
\begin{equation}
\label{eq:helstromobservable}
\Gamma=q\rho_{+}^{\otimes M}-(1-q)\rho_-^{\otimes M}\;.
\end{equation}
The state $\rho_{\pm}$ is guessed according to the sign of the measurement output. Note that, in some scenarios, $\Gamma$ will be either positive or negative meaning that the best strategy is always simply to guess the corresponding state. In practice, we implement the operator $\Gamma$, through a positive operator valued measure (POVM). A POVM is a set of positive operators $\{\Pi_k\geq0\}$, which sum to the identity $\sum_k\Pi_k=\mathbb{I}$. Using the optimal measurement, given in Eq.~\eqref{eq:helstromobservable}, and the definition of error probability in Eq.~\eqref{eq:errorprobabilitybasic}, we can verify that the Helstrom bound becomes $(1-\lVert\Gamma\rVert)/2$, as expected from Eq.~\eqref{eq:helstromdefinition}~\cite{bergou2010discrimination}. In general the POVM required to implement $\Gamma$ will correspond to a collective measurement. 


\subsection{Quantum Chernoff bound}
The quantum Chernoff bound (hereafter abbreviated to Chernoff bound) gives the error probability when we allow for collective measurements on asymptotically many copies of the unknown state. Specifically~\cite{audenaert2007discriminating,Nussbaum2009chernoff}
%
%
%
\begin{equation}
    \lim_{M\to\infty}\frac{-\text{log}(P_{\text{e,min},M})}{M}=-\text{log}(\kappa)\;,
\end{equation}
where
\begin{equation}
\label{eq:chernoffdefinition2}
\kappa=\min_{0\leq s\leq1}\text{Tr}[\rho_+^s\rho_-^{1-s}]\;,
\end{equation}
and $P_{\text{e,min},M}$ is the minimum error probability given $M$ copies of the probe state, i.e. $P_{\text{e,min},M}=P^{\text{H}}_\text{e}(M)$. Given $M\gg1$ copies of a quantum state, we can expect the error probability to scale as $P_{\text{e,min},M}\propto\text{e}^{-\epsilon_MM}$. We call this $\epsilon_M$ term the error exponent.

\section{Attainability of the quantum Chernoff bound}
\label{secres1}

\subsection{Example 1 - Simple qubit example}
\label{examplesimple}
We consider a very simple example first. We wish to discriminate between the two qubit states described by
\begin{equation}
\rho_{\pm}=\frac{1}{2}\bigg(\mathbb{I}_2\pm(1-v)(\sigma_3)\bigg)\;,
\end{equation}
where $\mathbb{I}_2$ is the $2\times2$ identity matrix, $\sigma_i$ denotes the $i$th Pauli matrix, and $0\leq v \leq 1$ describes the extent to which the state is mixed. Note that $v=0$ corresponds to a pure state and $v=1$ represents the maximally mixed state.

\subsubsection{Helstrom bound}
Let us consider the eigenvalues of $\rho_+^{\otimes M}-\rho_-^{\otimes M}$ as this will allow us to calculate the Helstrom bound. To do so, it is easiest to note that
\begin{equation}
\begin{split}
\rho_+&=\frac{1}{2}\big((2-v)\ket{0}\bra{0}+v\ket{1}\bra{1}\big)\\
\rho_-&=\frac{1}{2}\big(v\ket{0}\bra{0}+(2-v)\ket{1}\bra{1}\big)\;.
\end{split}
\end{equation}
We can then see that
\begin{equation}
\begin{split}
\rho_+^{\otimes M}&=\frac{1}{2^M}\sum_i(2-v)^{f_0(i)}v^{f_1(i)}\ket{\psi_i}\bra{\psi_i}\\
\rho_-^{\otimes M}&=\frac{1}{2^M}\sum_iv^{f_0(i)}(2-v)^{f_1(i)}\ket{\psi_i}\bra{\psi_i}\;,
\end{split}
\end{equation}
where the sum is over all $2^M$ possible states and $f_0(i)$ represents the number of times the state $\ket{0}$ is involved in the tensor product generating the state $\ket{\psi_i}$. Now we can see that the term $\norm{q\rho_+^{\otimes M}-(1-q)\rho_-^{\otimes M}}$, evaluated at $q=1/2$, is given by
\begin{equation}
\label{eqhelequal}
\frac{\norm{\rho_+^{\otimes M}-\rho_-^{\otimes M}}}{2}=\frac{1}{2^{M}}\sum_{i=0}^{\lfloor M/2\rfloor}\binom{M}{i}\big((2-v)^{M-i}v^{i}-v^{M-i}(2-v)^{i}\big)\;,
\end{equation}
which allows for easy evaluation of the Helstrom bound and the quantity $\epsilon_M$. 


\begin{figure*}[t]
\includegraphics[width=0.95\textwidth]{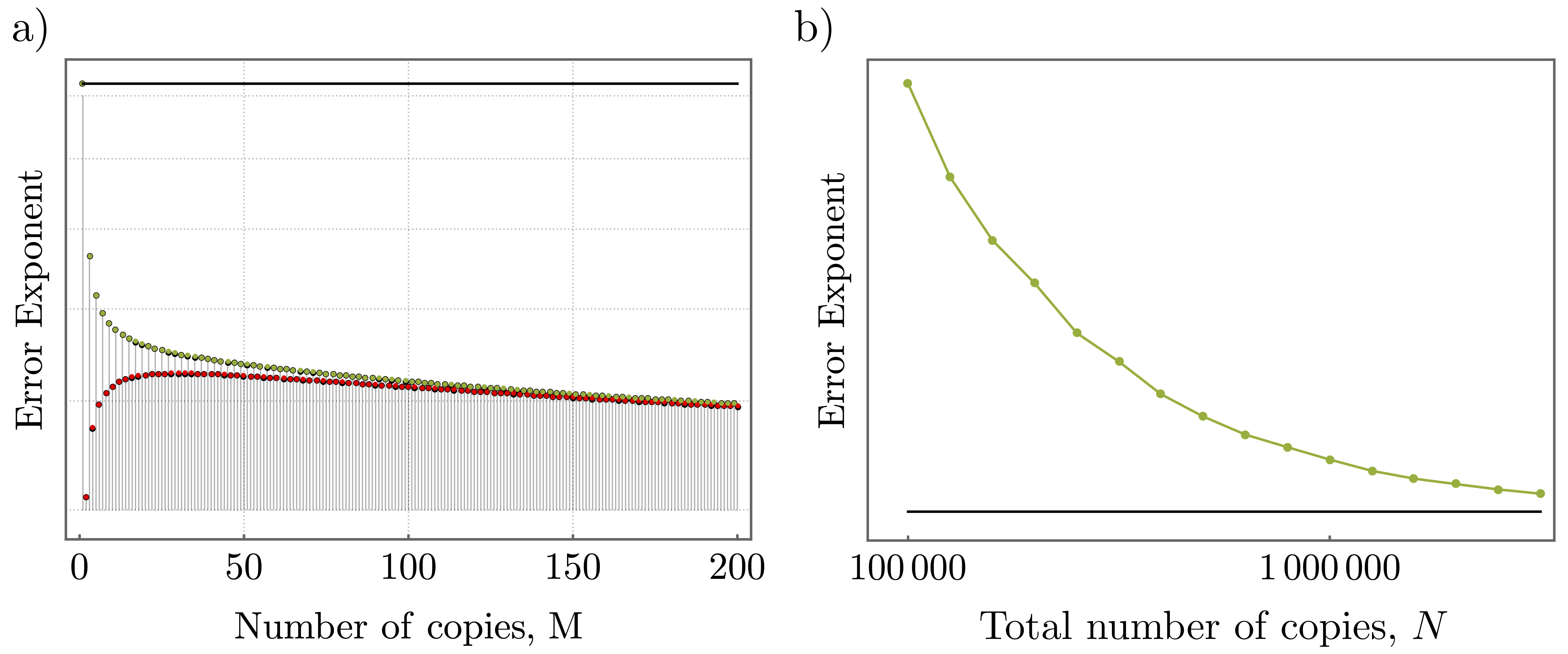}
\caption{\textbf{Error exponent attained by repeatedly implementing the Helstrom measurement.} a) The black line represents the error exponent given by the Chernoff bound. The green (odd $M$) and red (even $M$) points represent the true error probability obtained when repeatedly implementing the $M$-copy Helstrom bound for $N=1\times10^6$. The grey lines and black circles correspond to the lower bound on the error exponent given by Eq.~\eqref{errorExpRepeatLB}. b) The error exponents obtained when repeatedly implementing the single-copy Helstrom measurement are shown in green. The Chernoff bound is shown in black. Both plots use $v=0.8$.}
\label{fig:HelRep}
\end{figure*}

\subsubsection{Helstrom bound with unequal priors}
For this example, we can also calculate the Helstrom bound with unequal priors
\begin{equation}
\begin{split}
&\norm{q\rho_+^{\otimes M}-(1-q)\rho_-^{\otimes M}}=\\& \frac{1}{2^{M}}\sum_{i=0}^{ M}\binom{M}{i}\abs{\big(qv^{i}(2-v)^{M-i}-(1-q)(2-v)^{i}v^{M-i}\big)}\;.
\end{split}
\end{equation}
Making an analytic comparison of the Helstrom bound with equal and unequal priors is difficult due to the presence of the absolute value function in the above expression. The term inside the absolute function is a decreasing function of $i$, so we can write
\begin{equation}
\label{eq:eg1unequalq}
\begin{split}
&\norm{q\rho_+^{\otimes M}-(1-q)\rho_-^{\otimes M}}=\\& \frac{1}{2^{M}}\sum_{i=0}^{T}\binom{M}{i}\big(qv^{i}(2-v)^{M-i}-(1-q)(2-v)^{i}v^{M-i}\big)\\ +&\frac{1}{2^{M}}\sum_{i=T+1}^{M}\binom{M}{i}\big((1-q)(2-v)^{i}v^{M-i}-qv^{i}(2-v)^{M-i}\big)\;,
\end{split}
\end{equation}
where $T$ represents the point where the function inside the absolute sign changes from positive to negative. This function is positive when
\begin{equation}
i\leq\frac{\text{log}(q/(1-q))+M\text{log}((2-v)/v)}{2\text{log}((2-v)/v)}\;.
\end{equation}
We see that when $q=1/2$, this reduces to $i\leq M/2$ as expected. 
\subsubsection{Chernoff bound}

For this example the Chernoff bound is optimised for $s=1/2$ and can be calculated as
\begin{equation}
\label{eqChernoffsimple}
\kappa=[(2-v)v]^{1/2}\;,
\end{equation}
which allows the quantity $\epsilon_\infty$ (i.e. the asymptotic error exponent) to be calculated.

\subsubsection{Exponent attained by repeatedly performing the Helstrom measurement}
\label{sssecHelrepeat}
In this section, we consider the scenario where we have $N$ total copies of the quantum state and we perform the $M$-copy Helstrom measurement $N/M$ times, assuming $N\gg M$. A decision on which state is present can then be made by majority vote. Physically this corresponds to the situation where we receive either the state $\rho_+^{\otimes N}$ or $\rho_-^{\otimes N}$, however we are only capable of performing an entangling measurement on $M\ll N$ copies. It should be noted that repeatedly performing the Helstrom measurement is not necessarily the optimal strategy when allowing for $M$-copy entangling measurements on $N$ total copies of the state. As such, we do not necessarily expect the error exponent from this strategy to tend to the Chernoff exponent in the asymptotic limit.

The details of how the error exponent is calculated for this measurement strategy are presented in appendix~\ref{apen:HelRep}. In brief, the outcome of each Helstrom measurement is treated as a binomial random variable. We then make a final decision depending on which outcome, $\rho_+^{\otimes M}$ or $\rho_-^{\otimes M}$, occurs more frequently. This allows us to place bounds on the error exponent using the classical Chernoff bound. This is shown in Fig.~\ref{fig:HelRep}. Interestingly we see that the single-copy Helstrom measurement appears to be sufficient to saturate the Chernoff bound in the asymptotic limit. In Fig.~\ref{fig:HelRep} b), we show the convergence of this strategy to the Chernoff bound. This is consistent with other known results, see for example Ref.~\cite{acin2005multiple}, where repeated single-copy measurements and a  unanimity vote strategy was shown to saturate the Chernoff bound asymptotically. It is an interesting question, however not one that we address here, as to the optimal $M$-copy measurement for minimising the asymptotic error exponent. It may be that, for example, an optimised two-copy measurement can approach the Chernoff bound faster.

\begin{figure*}[t]
\includegraphics[width=0.95\textwidth]{fig_helstrom_repeat_ex2.png}
\caption{\textbf{Error exponent attained by repeatedly implementing the Helstrom measurement for the example in section~\ref{subsec-eg2}.} a) The black line represents the error exponent given by the Chernoff bound. The green (odd $M$) and red (even $M$) points represent the true error probability obtained when repeatedly implementing the $M$-copy Helstrom bound for $N=70,000$. The grey lines and black circles correspond to the lower bound on the error exponent given by Eq.~\eqref{errorExpRepeatLB}. b) The error exponents obtained when repeatedly implementing the single-copy Helstrom measurement are shown in green. The Chernoff bound is shown in black. Both plots use $v=0.8$ and $\alpha=\pi/4$.}
\label{fig:HelRepex2}
\end{figure*}
\subsection{Example 2 - More complex qubit example}
\label{subsec-eg2}
We now consider a second simple, but slightly more complex, example. We wish to discriminate between the two qubit states described by
\begin{equation}
\label{eq:complexqubit}
\rho_{\pm}=\frac{1}{2}\bigg(\mathbb{I}_2+(1-v)\big(\sigma_3\text{cos}(\alpha)\pm\sigma_1\text{sin}(\alpha)\big)\bigg)\;,
\end{equation}
where $\mathbb{I}_2$ is the $2\times2$ identity matrix, $\sigma_i$ denotes the $i$th Pauli matrix, and $v$ and $\alpha$ describe the extent to which the state is mixed and the angle between the two states in Hilbert space respectively. As before $v=0$ corresponds to a pure state and $v=1$ represents the maximally mixed state.

\subsubsection{Chernoff bound}
For this more complex example, we can also verify analytically that the Chernoff bound is optimised at $s=1/2$, giving
\begin{equation}
\label{eqChernoffeg2}
\kappa=\bigg(\sqrt{v(2-v)}-\text{cos}(\alpha)^2(\sqrt{v(2-v)}-1)\bigg)\;.
\end{equation}

\subsubsection{Pure state Helstrom bound}
For pure states, $\rho_+$ and $\rho_-$ only have one eigenvector and eigenvalue. Writing $\rho_\pm(v=0)=\ket{\psi_\pm}\bra{\psi_\pm}$, these are given by
\begin{equation}
\ket{\psi_\pm}=\begin{pmatrix}
\text{cos}(\alpha/2)\\
\pm\text{sin}(\alpha/2)
\end{pmatrix}\;.
\end{equation}
For pure states the Helstrom bound can be simplified as
\begin{equation}
\begin{split}
P_e^\text{H}(M)&=\frac{1}{2}\bigg(1-\sqrt{1-\abs{(\bra{\psi_+}^{\otimes M})(\ket{\psi_-}^{\otimes M})}^2}\bigg)\\
&=\frac{1}{2}\bigg(1-\sqrt{1-\abs{\bra{\psi_+}\ket{\psi_-}}^{2M}}\bigg)
\end{split}\;,
\end{equation}
and for our problem
\begin{equation}
\label{eqHelstromeg2}
(\bra{\psi_+}\ket{\psi_-})^2=\text{cos}(\alpha)^2\;.
\end{equation}

\subsubsection{Error exponent obtained by repeatedly performing the Helstrom measurement}
For mixed states, we are able to compute the Helstrom bound numerically. This enables us to use the same approach as in section~\ref{sssecHelrepeat}, to evaluate how the Chernoff bound is approached as we repeatedly implement the Helstrom measurement. As is shown in Fig.~\ref{fig:HelRepex2} a), in this case the lower bound from appendix~\ref{apen:HelRep} is not as close to the true error exponent as in Fig.~\ref{fig:HelRep}. However, as is shown in in Fig.~\ref{fig:HelRepex2} b), once again we see that repeatedly implementing the single-copy Helstrom bound is sufficient to asymptotically approach the Chernoff bound.

\section{Requirements to saturate the $M$-copy Helstrom bound}
\label{secres2}
We now wish to examine the requirements for saturating the $M$-copy Helstrom bound for any finite value of $M$. In Table~\ref{T2} we have summarised known results on when LOCC measurements can saturate the $M$-copy Helstrom bound. However, there exist other measurements which are more powerful than LOCC measurements but less powerful than a collective measurement on all $M$ copies. For the example considered in section~\ref{subsec-eg2}, we have previously shown that, with equal priors ($q=1/2$), the two-copy Helstrom bound can be saturated with LOCC measurements~\cite{conlon2023discriminating}, i.e. 2-copy collective measurements are not required. From a practical viewpoint, it is interesting to check whether this trend holds in general, i.e. are $M-1$-copy collective measurements sufficient to saturate the $M$-copy Helstrom bound? Or are there scenarios where, for example, different collective measurements on $M/2$ copies of the state can saturate the $M$-copy Helstrom bound?


Throughout this section, we will consider a class of measurements on $M$ copies of the quantum state, where the measurement is broken into 2 stages. In the first stage we perform a collective measurement on $M_1$ copies of the quantum state and in the second stage we perform a collective measurement, which depends on the results of the first measurement, on the remaining $M_2=M-M_1$ copies of the quantum state. We denote the error probability for this strategy as $P_{M_1,M_2}^{M}$. By introducing this family of measurements we are effectively constructing a bipartite quantum state discrimination problem, where the first party has $M_1$ copies of the quantum state and the second party has $M_2$ copies. We allow one-way LOCC measurements for the two parties. In this setting it is known that for certain examples LOCC measurements are insufficient to saturate the Helstrom bound - a phenomena known as non-locality without entanglement~\cite{bennett1999quantum,bennett1999unextendible}.

\begin{figure*}[t]
\includegraphics[width=\textwidth]{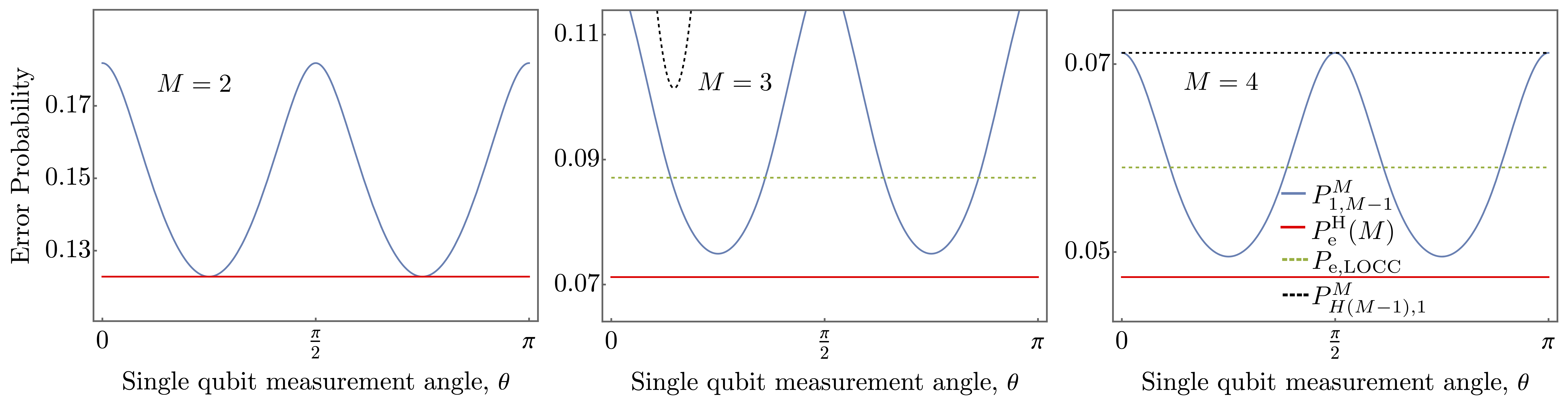}
\caption{\textbf{Attainability of the $M$-copy Helstrom bound.}  From left to right we compare the $M$-copy Helstrom bound to several measurements using only at most $M-1$-copy collective measurements, for $M=2$, $M=3$ and $M=4$ respectively. All plots correspond to $v=0.1$ and $\alpha=\pi/4$. The Helstrom bound and LOCC error probability are not a function of the single qubit measurement angle as this measurement angle refers to the two-step strategies. For the $P^M_{1,M-1}$ strategies, the \mbox{$x$-axis} (single qubit measurement angle) refers to the measurement angle in the first stage, and for the $P^M_{H(M-1),1}$ strategies, the \mbox{$x$-axis} refers the the measurement angle for the second stage. Note that in the latter case, the optimal measurement angle in the second stage depends on the measurement result in the first stage. Here for clarity, we only plot the total error probability as a function of one of these measurement angles.}
\label{fig:fig_helstrom}
\end{figure*}

It turns out that it is quite simple to get analytic results for one particular type of this two stage measurement, the $P_{1,M-1}^{M}$ measurement. We can perform a single qubit measurement, followed by a (possibly adaptive) collective measurement on the remaining copies. In general, if we first perform a single qubit measurement, we need only consider the measurements defined by the projectors $\Pi_0=\ket{\psi}\bra{\psi}$ and $\Pi_1=\ket{\psi^\perp}\bra{\psi^\perp}$ where
\begin{equation}
\label{eq:meass}
\ket{\psi}=\text{cos}(\phi)\ket{0}+\text{sin}(\phi)\ket{1}\;,
\end{equation} 
and $\bra{\psi}\ket{\psi^\perp}=0$~\cite{higgins2011multiple}. Denote the outcome of this measurement as $D_i$, so that $D_i$ can correspond to either $\Pi_0$ or $\Pi_1$ clicking. Given the measurement outcome $D_i$, the posterior probability after performing this single qubit measurement is given by
\begin{equation}
\label{eq:JM1}
q_{i}=\frac{\text{Pr}[D_i|\rho_{+},\phi_i]}{2\text{Pr}[D_i|\phi_i]}\;,
\end{equation}
where the factor of 2 in the denominator is a result of the initial prior probability being equal to $1/2$, $\phi_i$ represents the measurement angle in Eq.~\eqref{eq:meass}, and $\text{Pr}[D_i|C]$ is the probability of measurement outcome $D_i$ given the conditions $C$. Depending on which detector clicks, we therefore get a different posterior probability. This posterior probability then becomes our prior probability for the collective measurement on the remaining $M-1$ copies. Hence, the error probability in this scheme is given by
\begin{equation}
P_{1,M-1}^{M}=\text{Tr}[\rho\Pi_1]P_\text{e}^{\text{H}}(M-1,q_1)+\text{Tr}[\rho\Pi_0]P_\text{e}^{\text{H}}(M-1,q_0)\;,
\end{equation}
where we have made the dependence of the Helstrom bound on the prior probability explicit. This approach will allow us to compare different measurement strategies to the $M$-copy Helstrom bound.

\subsection{An example where the $M$-copy Helstrom bound can be saturated without collective measurements on all $M$ copies}
Let us first return to the example considered in section~\ref{examplesimple}. As the two states in this example are simultaneously diagonalisable, the Helstrom bound can be saturated through the maximum likelihood measurement, see Ref.~\cite{audenaert2014upper} and appendix~\ref{apen:simdiag}.

\subsection{An example where the $M$-copy Helstrom bound requires collective measurements on all $M$ copies to be saturated}
Let us now return to the example considered in section~\ref{subsec-eg2}. For the measurement described in Eq.~\eqref{eq:meass}, and the unknown states given by Eq.~\eqref{eq:complexqubit} with $\rho=(\rho_++\rho_-)/2$, we can calculate
\begin{equation}
\begin{split}
   \text{Tr}[\rho\Pi_0]&=\frac{1}{2}\big(1+(1-v)\text{cos}(\alpha)\text{cos}(2\phi)\big)\\
   \text{Tr}[\rho\Pi_1]&=\frac{1}{2}\big(1-(1-v)\text{cos}(\alpha)\text{cos}(2\phi)\big)\;.
   \end{split}
\end{equation}
and
\begin{equation}
\begin{split}
   \text{Tr}[\rho_+\Pi_0]&=\frac{1}{2}\big(1+(1-v)\text{cos}(\alpha-2\phi)\big)\\
   \text{Tr}[\rho_+\Pi_1]&=\frac{1}{2}\big(1-(1-v)\text{cos}(\alpha-2\phi)\big)\;.
   \end{split}
\end{equation}
This allows us to compare $P_{1,M-1}^{M}$ and $ P_{\text{e}}^{\text{H}}(M)$. We can also compare with the strategy from LOCC measurements only, which we denote $ P_{\text{e,LOCC}}$. For $M=2$, the $P_{1,M-1}^{M}$ strategy corresponds to LOCC measurements.

 Interestingly we see that for $M=2$ and equal prior probabilities of the two states ($q=1/2$), LOCC measurements are sufficient to saturate the two-copy Helstrom bound. In appendix~\ref{apen:2copyLOCCmeas}, we prove this for all $\alpha$ and $v$. However, for $M=3$ and $M=4$, two- and three-copy collective measurements respectively, do not appear to be sufficient. This is shown in Fig.~\ref{fig:fig_helstrom}. Another strategy which we can analytically obtain results for, is to perform the $M-1$-copy Helstrom measurement on the first $M-1$ copies and then perform an optimised single qubit measurement on the final copy. Note however, that in this case the $M-1$-copy Helstrom measurement may not be the optimal measurement to minimise the overall error probability. Also note that the optimal measurement on the final copy will simply be the Helstrom measurement for the updated posterior probability. We denote the error probability for this strategy as $P_{H(M-1),1}^{M}$. For $M=3$ and $M=4$ these results are also shown in Fig.~\ref{fig:fig_helstrom}. Interestingly, even the LOCC measurements can outperform this particular measurement. Note that as there are two possible outcomes for the Helstrom measurement there are two single qubit measurement angles to be optimised in the final measurement stage. However, due to the symmetry of the problem, and for clarity of presentation, we simply plot twice the error probability for one of these outcomes as a function of the corresponding measurement angle. It should also be noted that for the $M=4$ case shown in Fig.~\ref{fig:fig_helstrom}, the optimal strategy is simple to make a guess based on the measurement of the first 3 copies. This is reflected in the fact that the $P_{H(3),1}^{4}$ strategy is not a function of the single qubit measurement angle.

Finally, it should be noted that other measurement strategies may be optimal. For example, for $M=4$, perhaps two consecutive two-copy collective measurements are sufficient. Or an optimised three-copy measurement first and then a single-copy measurement. In the following section we make an attempt at optimising such measurements.

\subsection{Variational quantum circuits for finding optimal bipartite LOCC measurement strategies}
\label{subsecPQC}

From the above, we know that the $P_{1,2}^{3}$ strategy does not saturate the 3-copy Helstrom bound. Similarly, $P_{1,3}^{4}$ does not saturate the 4-copy Helstrom bound. In both of these scenarios, LOCC measurements also do not saturate the Helstrom bound. To investigate the attainability of the $M$-copy Helstrom bound more generally, we shall now use parameterised quantum circuits to investigate the optimal $P_{M-N,N}^{M}$ measurement. In general, for this strategy, the error probability can be written as
\begin{equation}
\label{eq:JM2}
P_{M-N,N}^{M}=\sum_i\text{Tr}[\rho^{\otimes (M-N)}\Pi_i]P_\text{e}^{\text{H}}(N,q_i)\;,
\end{equation}
where $q_i$ is a function of $\Pi_i$. We can think of an \mbox{$(M-N)$} qubit quantum circuit as an $2^{M-N}$ element POVM. Hence through parameterised quantum circuits, we can numerically find the optimal $P_{M-N,N}^{M}$ strategy.

To give a simple example: for the $P_{3,1}^{4}$ strategy, we use a circuit consisting of layers of single qubit unitary operations followed by alternating CNOT gates. Each three qubit circuit corresponds to a measurement on three-copies of the unknown state. There are 8 possible measurement outcomes each of which gives rise to a different posterior probability. Based on these posterior probabilities we can then calculate the Helstrom bound in each of these 8 cases, which summed and weighted by the probability of the corresponding outcome gives us $P_{3,1}^{4}$. We repeated this increasing the number of circuit layers, until the error probability converged to the final value. Note that the circuit described above can be used to optimise the $P_{3,M-3}^{M}$ strategy for any $M\geq4$. More details are provided in the Methods section.

\begin{figure*}[t]
\includegraphics[width=0.95\textwidth]{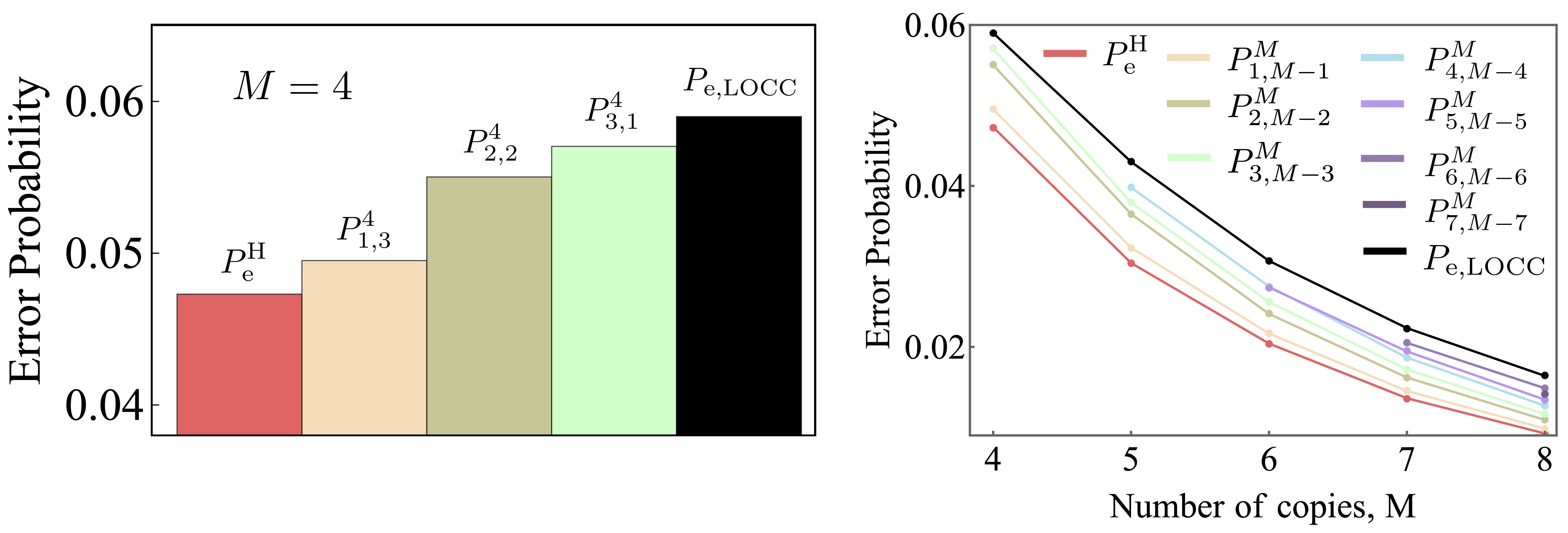}
\caption{\textbf{Attainability of the $M$-copy Helstrom bound with different LOCC measurement strategies.}  On the left we plot the probability of error for different measurement strategies optimised using parameterised quantum circuits. We see that no measurement strategy which does not use collective measurements on all 4-copies simultaneously can saturate the Helstrom bound. On the right we extend these results beyond $M=4$. $v=0.1$ and $\alpha=\pi/4$ for this figure.}
\label{fig:fig_optCircs}
\end{figure*}

 These results are shown in Fig~\ref{fig:fig_optCircs}. It is clear that in this example none of these methods can saturate the Helstrom bound. On the right of Fig~\ref{fig:fig_optCircs}, we extend these results to more than four copies of the unknown states. These results appear to suggest that for this family of states the $M$-copy Helstrom bound requires $M$-copy collective measurements to be saturated. However, this comes with the caveat that the optimisation of our parameterised quantum circuits may not give the true optimal answer. These results do indicate that if it is possible to perform collective measurements on some but not all $M$ copies of the unknown state, the best strategy is to perform the collective measurement on the final copies of the unknown state. It is also worth noting that the $P_{7,1}^{8}$ strategy appears to outperform the $P_{6,2}^{8}$ strategy, which is inconsistent with the general ordering of the rest of our results. However, this may simply be an issue with our optimisation process rather than a meaningful physical effect.


As is noted in Table~\ref{T2}, for pure states ($v=0$) the Helstrom bound can be saturated by LOCC measurements~\cite{brody1996minimum,ban1997accessible,acin2005multiple}. We verify this numerically for our problem, confirming that the LOCC strategy can saturate the 4-copy Helstrom bound up to 5 significant figures.


Finally, we wish to examine if there is a general hierarchy between two different classes of measurement: LOCC measurements and non-adaptive collective measurements on less than all $M$ copies of the unknown state. For this we compare LOCC measurements to two-copy entangling measurements, where we perform the same two-copy measurement on all pairs of the unknown state. Our initial simulations suggest that LOCC measurements are more powerful, however a more comprehensive study is required to verify this in a range of different settings.

\subsection{Numerical search for other measurements saturating the Helstrom bound}
\begin{figure}[t]
\includegraphics[width=0.47\textwidth]{fig_opt_meas_PPT.png}
\caption{\textbf{Attainability of the $M$-copy Helstrom bound with different PPT measurement strategies.}  As in Fig.~\ref{fig:fig_optCircs}, however now we consider PPT measurements between a bipartition $M_1, M_2$ as opposed to one-way LOCC measurements.}
\label{fig:fig_optCircs_PPT}
\end{figure}

\begin{figure*}[t]
\includegraphics[width=0.9\textwidth]{fig_PPT_num.png}
\caption{\textbf{Attainability of the Helstrom bound with different PPT measurement strategies.}  In both figures the dashed line corresponds to a $y=x$ line. The left figure shows 10,000 randomly sampled problems discriminating between four copies of three known qubit states. The red and blue data points correspond to $P_{2,2}^{4, \text{ PPT}}$ and $P_{1,3}^{4, \text{ PPT}}$ measurements respectively. The inset shows the difference between the Helstrom measurement and the corresponding PPT measurement. The x-axis of the inset is the same as the main figure. The right figure shows 10,000 randomly sampled problems discriminating between three quantum states in a 16 dimensional Hilbert space. }
\label{fig:fig_PPT_num}
\end{figure*}
Above we effectively considered a bipartite state discrimination problem with one-way LOCC measurements allowed between the two parties. To further examine the attainability of the Helstrom bound we now consider a broader range of possible measurements. For discriminating multipartite quantum states, several studies have investigated the role of measurements beyond simply LOCC operations, namely separable and positive partial transpose (PPT) measurements~\cite{yu2014distinguishability,bandyopadhyay2015limitations,gungor2016entanglement,bandyopadhyay2018optimal,bandyopadhyay2021entanglement}. Consider a $d$-dimensional quantum state shared between two parties, Alice and Bob. We shall denote the dimensions of the quantum state held by Alice and Bob as $d_A$ and $d_B$ respectively. Then a POVM $\{\Pi_k\}$ is separable if each $\Pi_k$ can be written as 
\begin{equation}
\Pi_k=\Pi_{k,A}\otimes\Pi_{k,B}\;,
\end{equation}
where $\Pi_{k,A}$ and $\Pi_{k,B}$ represent valid POVMs on Alice and Bob’s respective Hilbert spaces. It is clear that every LOCC measurement can be written as a separable measurement, however the converse is not necessarily true.

PPT measurements are a broader class of measurement again. The partial transpose operation acting on $A\otimes B$ takes the transpose of one of the subsystems, $A^T\otimes B$. PPT measurements are POVMs that remain positive under the action of partial transpose. The set of all PPT measurements is a closed convex cone, meaning that the optimal PPT measurement for state discrimination can be found via a semi-definite program (SDP)~\cite{cosentino2013positive,zhu2024entanglement}. Interestingly, SDPs appear in many other branches of quantum information~\cite{conlon2021efficient,kanitschar2024practical,kossmann2024semidefinite,lai2024semidefinite,tavakoli2024semidefinite}. 

Every separable measurement is a PPT measurement but the converse is not true. Therefore, we can construct a broader class of measurements where we split the $M$-copy quantum state to be discriminated into two groups of $M_1$ and $M_2$ copies, i.e. $d_A=d^{\otimes M_1}$ and $d_B=d^{\otimes M_2}$, where $d$ is the dimension of a single copy of the quantum state. We then allow PPT measurements between these two groups, and we denote the error probability for this strategy as $P_{M_1,M_2}^{M, \text{ PPT}}$. Owing to the symmetry across the partition $P_{M_1,M_2}^{M, \text{ PPT}}=P_{M_2,M_1}^{M, \text{ PPT}}$. The results for this class of measurement are shown in Fig.~\ref{fig:fig_optCircs_PPT}. As in the previous section we see that even PPT measurements among any bipartition is not sufficient to saturate the $M$-copy Helstrom bound. However, it should be noted that, as expected, the PPT measurement performs better than the corresponding one-way LOCC measurement.

\subsection{Random examples to investigate the attainability of the Helstrom bound}

Given any multi-qubit or multi-copy quantum state from a known set of quantum states to be discriminated we can efficiently compute the optimal PPT measurement given any partition. This is true even when discriminating more than two unknown quantum states. The Helstrom bound for discriminating many different states can also be written as an SDP~\cite{watrous2018theory}. This enables us to compare the PPT measurements to the optimal Helstrom measurement for many different random examples. This is done in Fig.~\ref{fig:fig_PPT_num} for two different settings. In the left figure we compare PPT measurements corresponding to the $P_{2,2}^{4, \text{ PPT}}$ and $P_{1,3}^{4, \text{ PPT}}$ strategies for discriminating between four copies of three known qubit states, i.e we distinguish between $\rho_1^{\otimes4}$, $\rho_2^{\otimes4}$ and $\rho_3^{\otimes4}$. In this example, we find that in some situations the PPT measurement equals the  Helstrom measurement to within the numerical accuracy of the SDP. In the right figure, we show the same, however now we consider discriminating between three unknown 16 dimensional quantum states. In this case, we do not find any examples where the PPT measurement matches the performance of the Helstrom measurement. However, this difference may be attributable to the fact that a 16 dimensional Hilbert space is a lot richer than the space formed by taking the tensor product of four qubit states. Note that for these random examples we do not assume that the states are equally likely, i.e. we do not assume a uniform prior distribution.

\section{Conclusion}
\label{secconc}
We have obtained analytic expressions for the Helstrom bound and Chernoff bound for several simple qubit examples. Based on this we have evaluated the attainability of the asymptotic limits in quantum state discrimination. In particular we examined whether repeatedly implementing the $M$-copy Helstrom measurement is sufficient to approach the Chernoff bound asymptotically. In this setting we found that the single-copy Helstrom measurement performed best. However, this comes with the caveat that for repeatedly implementing the same measurement other measurements will be optimal and one could expect that some other repeated entangling measurement will outperform the single-copy measurement. However, finding this optimal measurement is left as an open quesiton.

We have also investigated the conditions necessary to saturate the Helstrom bound. In particular we have investigated a family of LOCC measurements which generate entanglement on less than all $M$ copies of the unknown state. We did not find any instances where this class of measurement saturates the Helstrom bound. However, there does appear to exist a hierarchy within this family of measurements. Similar results were obtained for a family of PPT measurements. However, as is shown in Fig.~\ref{fig:fig_PPT_num}, the attainability of the Helstrom bound depends strongly on the structure of the problem. Further investigation is needed to find the conditions under which the Helstrom bound is attainable without entanglement across all available modes.


Finally, it could be useful to extend our results to other settings, particularly the discrimination of Gaussian states~\cite{lloyd2008enhanced,tan2008quantum,bradshaw2021optimal}, or the use of other figures of merit, such as asymmetric state discrimination~\cite{karsa2020quantum,pereira2023analytical} or unambiguous state discrimination~\cite{ivanovic1987differentiate,dieks1988overlap}. Given the strong connection between quantum classification and state discrimination~\cite{banchi2021generalization}, it would be interesting to evaluate the performance of the family of measurements introduced here to this and related fields, such as the learning of quantum processes~\cite{seif2024entanglement}.

\section{Methods}
The parametrized quantum circuits comprise alternating layers of $U(2)$ unitary gates on every qubit and CNOT gates between pairs of adjacent qubits, shown in Fig.~\ref{fig:fig_circ_diagram}. The CNOTs are in brickwork pattern with closed (i.e. wrap-around) boundary conditions at the ends of the qubit chain, such that on an $N$-qubit circuit, a minimum of $\sim N/2$ CNOT layers suffice to generate an entanglement light cone spanning all qubits. Each $U(2)$ gate is parametrized by $3$ angles. We minimized the error probability (see Eqs.~\eqref{eq:JM1} and \eqref{eq:JM2} of main text) over all $U(2)$ gate angles, using a standard optimization algorithm (L-BFGS-B) with basin hopping. Our calculations employed 50 hops with a maximum of 500 iterations each. We checked for convergence by increasing the number of CNOT layers in the parametrized quantum circuits (up to ~14), and number of optimization hops/iterations, to ensure that the best-found error probability is stable.

\begin{figure}[t]
\includegraphics[width=0.47\textwidth]{fig_circuit_diag.png}
\caption{\textbf{Circuit structure used to find optimal LOCC measurements in section~\ref{subsecPQC}.} An example of the circuit structure used for 6 qubits with 3 layers is shown. Blue boxes represent single qubit operations and purple boxes represent CNOT gates.}
\label{fig:fig_circ_diagram}
\end{figure}

\section*{Data availability}
Any data obtained for this project are available from the corresponding authors upon request.
\section*{Code availability}
Any code used for this project are available from the corresponding authors upon request.
\section*{Acknowledgements}
This research was funded by the Australian Research Council Centre of Excellence CE170100012. This research is supported by A*STAR C230917010, Emerging Technology and A*STAR C230917004, Quantum
Sensing. JSS acknowledges support from the UK NQTP and the EPSRC Quantum Technology Hub in Quantum Communications (Grant Ref.: EP/T001011/1).

\appendix

\section{Saturating the Helstrom bound for simultaneously diagonalisable states through the maximum likelihood measurement}
\label{apen:simdiag}
We consider discriminating between the following two states 
\begin{equation}
\rho_\text{1}=\begin{pmatrix}
\lambda_{1,1}&0&\hdots&0\\
0&\ddots&\hdots&\vdots\\
\vdots&\hdots&\ddots&0\\
0&\hdots&0&\lambda_{1,n}
\end{pmatrix}\quad\text{and}\quad\rho_\text{2}=\begin{pmatrix}
\lambda_{2,1}&0&\hdots&0\\
0&\ddots&\hdots&\vdots\\
\vdots&\hdots&\ddots&0\\
0&\hdots&0&\lambda_{2,n}
\end{pmatrix}\;.
\end{equation}
For simultaneously diagonalisable states the probability of error is given by~\cite{audenaert2014upper}
 \begin{equation}
 \label{Eq:helstrommaxlike}
1- \text{Tr}[\text{max}(q\rho_1,(1-q)\rho_2)]\;,
 \end{equation}
 where the maximum is taken entrywise in the basis that simultaneously diagonalises both matrices. This error probability can be achieved by the maximum likelihood measurement. For this we can consider two POVMs
 \begin{equation}
 \begin{split}
 \Pi_1&=\text{diag}(d_{1,1},d_{1,2},\hdots,d_{1,n})\qquad\text{and}\\
 \Pi_2&=\text{diag}(d_{2,1},d_{2,2},\hdots,d_{2,n})\;,
 \end{split}
 \end{equation}
 where $d_{1,i}=1$ if $\lambda_{1,i}\geq\lambda_{2,i}$ and 0 otherwise and $d_{2,i}=1$ if $\lambda_{2,i}>\lambda_{1,i}$. It is evident that this POVM gives a probability of error equal to that in Eq.~\eqref{Eq:helstrommaxlike}.
 
%
\section{Necessary and sufficient condition for a POVM to be saturate the Helstrom bound}
\label{apen:POVMopt}
Following Ref.~\cite{ban1997optimum}, we note that for distinguishing any two states $\rho_1$ and $\rho_2$ we can use a two outcome POVM, $\{\Pi_1,\Pi_2\}$. In order to saturate the Helstrom bound, this POVM must satisfy
\begin{equation}
\begin{split}
&\Pi_1(q\rho_1-(1-q)\rho_2)\Pi_2=0\qquad\text{and}\\
&q\rho_1\Pi_1+(1-q)\rho_2\Pi_2-p_j\rho_j\geq0\;,
\end{split}
\end{equation}
where $p_j$ is the prior probability of the unknown state being the state $\rho_j$.

\section{Calculating the error probability when repeatedly performing the $M$-copy Helstrom measurement}
\label{apen:HelRep}
We now analyse the measurement strategy discussed in section~\ref{sssecHelrepeat}, where we have $N$ copies of the quantum state and we perform the $M$-copy Helstrom measurement $N/M$ times ($N\gg M$). Eq.~\eqref{eqhelequal} gives an analytic expression for the error probability when the Helstrom measurement is implemented. From the symmetry of the problem being considered (equal prior probabilities) we can conclude that the Helstrom measurement is symmetric. Hence, the probability of each outcome of the Helstrom measurement is equal to $P^{\text{H}}_\text{e}(M)$. We then repeat this measurement $N/M$ times, which can be modelled as a binomial distribution. Let us consider the following strategy. We assign the value 1 to successful events and the value 0 to incorrectly identified density matrices. In this case the expected value from the sum of all the random variables ($X=\sum_{i=1}^{N/M}X_i$) is $\mu=(1-P^{\text{H}}_\text{e}(M))N/M>N/2M$, and we will incorrectly guess the unknown density matrix if the sum of the random variables has a value less than $N/2M$. The classical Chernoff bound for the sum of $n$ repeated binomially distributed random variables characterised by probability $p$ states that~\cite{arratia1989tutorial}
\begin{equation}
P(X\leq k)\leq \text{e}^{-nD(\frac{k}{n}||p)}\;,
\end{equation}
where
\begin{equation}
D(a||p)=a\text{log}\bigg(\frac{a}{p}\bigg)+(1-a)\text{log}\bigg(\frac{1-a}{1-p}\bigg)
\end{equation}
is the relative entropy. For our example $n=N/M$, $k=N/2M$ and $p=1-P^{\text{H}}_\text{e}(M)$. Substituting in these values we arrive at 
\begin{equation}
P^{\text{H,r}}_\text{e}(N)=P\bigg(X<\frac{N}{2M}\bigg)\leq \text{e}^{-\frac{N}{2M}\big(\text{log}\big(\frac{1}{2P^{\text{H}}_\text{e}(M)}\big)+\text{log}\big(\frac{1}{2(1-P^{\text{H}}_\text{e}(M))}\big)\big)}\;,
\end{equation}
where we use the superscript H,r to denote that this error corresponds to the repeated Helstrom measurement. Therefore we arrive at a lower bound on the error exponent using this strategy, 
\begin{equation}
\label{errorExpRepeatLB}
\epsilon_{N}^\text{H,r}\geq\frac{1}{2M}\bigg(\text{log}\bigg(\frac{1}{2P^{\text{H}}_\text{e}(M)}\bigg)+\text{log}\bigg(\frac{1}{2(1-P^{\text{H}}_\text{e}(M))}\bigg)\bigg)\;.
\end{equation}

\section{LOCC measurement saturating the Helstrom bound for two copies of the unknown state and equal prior probabilities}
\label{apen:2copyLOCCmeas}
The two copy Helstrom bound is given by
\begin{equation}
\label{eq:HelstromapenF}
P^{\text{H}}_\text{e}(2)=\frac{1}{2}\bigg(1-\left\lVert q\rho_+^{\otimes 2}-(1-q)\rho_-^{\otimes 2}\right\rVert\bigg)\;,
\end{equation}
with $\rho_\pm$ defined in Eq.~\eqref{eq:complexqubit}. It is easily verified that this becomes
\begin{equation}
P^{\text{H}}_\text{e}(2)=\frac{1}{4}\bigg(2-\sqrt{2}(1-v) \sqrt{3-(2-v)v+(1-v)^2\text{cos}(2\alpha)}\text{sin}(\alpha)  \bigg)\;.
\end{equation}

We now present a LOCC measurement strategy which saturates this bound. On the first copy of the unknown state, we implement the measurement described in Eq.~\eqref{eq:meass} with $\phi=\pi/4$. If $\Pi_0$ clicks in the first measurement the posterior probability (equivalently the prior probability for the second round) becomes
\begin{equation}
q_{1|0}=\frac{1}{2}\bigg(1+(1-v)\text{sin}(\alpha)\bigg)\;.
\end{equation}
This is calculated using Eq.~\eqref{eq:JM1}. Similarly if $\Pi_1$ clicks in the first measurement the posterior probability is given by
\begin{equation}
q_{1|1}=\frac{1}{2}\bigg(1-(1-v)\text{sin}(\alpha)\bigg)\;.
\end{equation}
Depending on which POVM element clicked we implement the corresponding optimal measurement in Eq.~\eqref{eq:helstromobservable}. As the probability of both POVMs clicking are equal to $1/2$, the LOCC error is then given by
\begin{equation}
P_{\text{e,LOCC}}=\frac{1}{2}(P^{\text{H}}_\text{e}(1,q_{1|0})+P^{\text{H}}_\text{e}(1,q_{1|1}))\;,
\end{equation}
where the two arguments in $P^{\text{H}}_\text{e}$ are the number of copies (i.e. the one remaining copy) and the prior probability. It is then easily verified that $P_{\text{e,LOCC}}$ equals the Helstrom bound in Eq.~\eqref{eq:HelstromapenF}.

\section{Quantities to investigate the error scaling in finite limit }
Finally, we wish to note that in the initial stages of this study we considered alternative methods of quantifying the attainability of the quantum Chernoff bound, which we present here for any interested reader. Essentially, we have investigated whether or not the quantum Chernoff bound is attainable given any finite number of copies of the unknown state. This amounts to asking whether the error probability provided by the Helstrom bound (when appropriately normalised for the number of copies of the unknown state being considered) ever reaches the error probability specified by the quantum Chernoff bound. However, the physical interpretation of these results is not so obvious as the Chernoff bound is exclusively an asymptotic quantity.

\begin{figure}[t]
\includegraphics[width=0.45\textwidth]{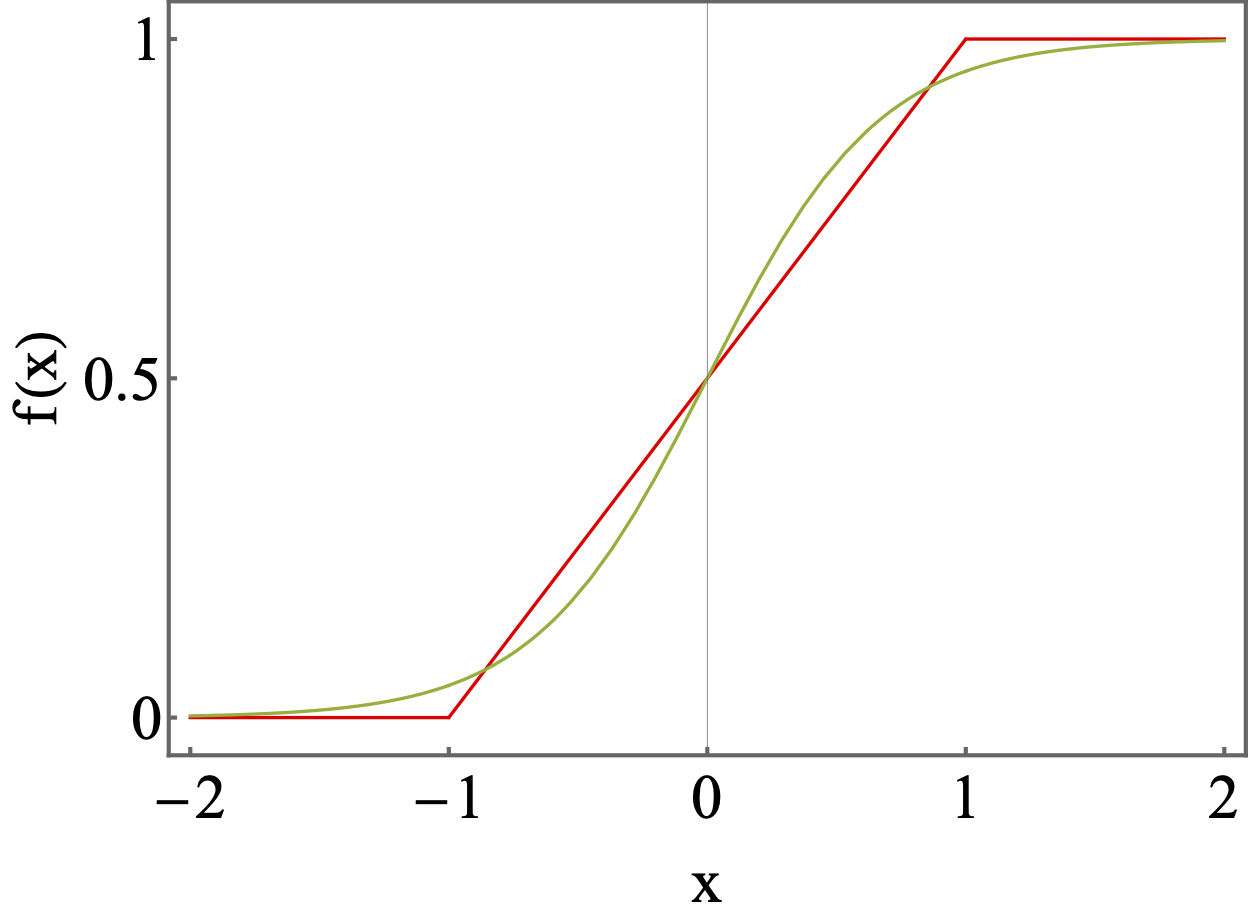}
\caption{\textbf{Difference between an attainable asymptotic limit and an unattainable asymptotic limit.} The function plotted in red, features a discontinuous gradient and as such the asymptotic value of this function can be attained for a finite argument. In green, we plot a sigmoid function, which has the same asymptotic limit as the red function, however, for no finite argument of the sigmoid function is the asymptotic limit attained.}
\label{fig:asympt}
\end{figure}

Our question is whether the error probability given by the Chernoff bound is attainable with any finite number of copies of the unknown state. To clarify what we mean by this, we have plotted two functions in Fig.~\ref{fig:asympt}. For one of these functions, the asymptotic limit of the function is attainable for finite arguments, for the other it is not. As the Chernoff bound specifically applies only in the asymptotic limit, it may be expected that the asymptotic limit cannot be reached for any finite argument, as is the case for many mathematical functions. Nevertheless, it is still worthwhile to investigate this, and related, questions. For instance, when investigating this question, we found interesting results about the rate at which the Chernoff bound is approached. 

In order to address this question we will need some method of normalising between the Helstrom bound (which applies in the finite limit) and the Chernoff bound (which applies in the asymptotic limit). Noting that
\begin{equation}
    \lim_{M\to\infty}\frac{-\text{log}(P_{\text{e,min},M})}{M}= \lim_{M\to\infty}-\text{log}(P_{\text{e,min},M}^{1/M})=-\text{log}(\kappa)\;,
\end{equation}
we can see that the appropriate normalisation to allow a comparison of error probabilities with different numbers of copies of the unknown state is to take the $M$-th root of the error probability. Furthermore, we have that
\begin{equation}
 \lim_{M\to\infty}P_{\text{e,min},M}^{1/M}=\kappa\;.
\end{equation}

In addition to comparing the error probabilities directly, it will prove instructive to compare error exponents.
For this purpose, let us define 
\begin{equation}
\begin{split}
        \epsilon_M&=-\frac{\text{log}(P_{\text{e}}^\text{H}(M))}{M}\qquad\text{and}\\
        \epsilon_\infty&=-\text{log}(\kappa)\;.
\end{split}
\end{equation}
Note that for distinguishing more than two states, the Chernoff bound can be off by a constant factor~\cite{audenaert2014upper}. However, for distinguishing two states this effect is not present. By comparing $\epsilon_M$ and $\epsilon_\infty$ we can examine how rapidly the ultimate limits on the error exponent can be approached.

When distinguishing quantum states we find situations where $\epsilon_1\geq\epsilon_2\geq\epsilon_3\hdots\geq\epsilon_\infty$, or equivalently $P_\text{e}(1)\leq P_\text{e}(2)^{1/2}\leq P_\text{e}(3)^{1/3}\hdots\leq \kappa$. At first glance this may appear surprising, however this is entirely in keeping with the results of Refs.~\cite{audenaert2007discriminating,Nussbaum2009chernoff}. Specifically Ref.~\cite{Nussbaum2009chernoff} proved
\begin{equation}
    \lim_{M\to\infty}\frac{\text{log}(P_{\text{e,min},M})}{M}\geq\text{log}(\kappa)\;,
\end{equation}
and Ref.~\cite{audenaert2007discriminating} derived
\begin{equation}
\frac{\text{log}(P_{\text{e,min},M})-\text{log}(q^s(1-q)^{1-s})}{M}\leq\text{log}(\kappa)\;,
\end{equation}
which implies
\begin{equation}
\frac{\text{log}(P_{\text{e,min},M})}{M}\leq\text{log}(\kappa)\;.
\end{equation}
Crucially, the first inequality holds only in the asymptotic limit, whereas the second holds for any value of $M$. Therefore, as expected we have $P_\text{e}(M)^{1/M}\leq \kappa$ and $\epsilon_M\geq\epsilon_\infty$. Note that it is of course possible to simply perform a worse single-copy measurement so that $P_\text{e}(1)= P_\text{e}(2)^{1/2}$. However, here we are only interested in optimal measurements given a finite number of copies. In this context, because of the additional $\text{log}(q^s(1-q)^{1-s})$ term, it is evident that we will never have equality between $P_\text{e}(M)^{1/M}$ and $\kappa$. In what follows we shall consider a quantity which accounts for this additional term, allowing up to investigate asymptotic properties of quantum state discrimination.

Let us define 
\begin{equation}
\label{eq:ratioPe}
R(M)=\frac{P^{\text{H}}_\text{e}(M)^{1/M}}{(q^{s^*}(1-q)^{(1-s^*)})^{1/M}\kappa}\;.
\end{equation}
where $s^*$ is the value of $s$ minimising $\kappa$ in Eq.~\eqref{eq:chernoffdefinition2}. We will also define

\begin{equation}
\hspace{-1cm}
    \delta(M,N)=R(M)-R(N)\;,
\end{equation}
where $M>N$.
These two quantities will let us examine how the error probability changes as a function of $M$ and whether the asymptotic error probability can ever be reached. We wish to investigate whether $R(M)$ can exactly equal $\lim_{M\to\infty}R(M)=1$ for any finite $M$. Essentially this amounts to showing that $R(M)$ never plateaus. If this is the case, then we will know that the error exponent specified by the Chernoff bound is not attainable using collective measurements on any finite number of copies of the unknown state. It is only attainable in the asymptotic limit. Alternatively, we could show that $\delta(M,N)$ never goes to 0 for specific problems.

\subsection{Finite size comparison of bounds in example 1}
Given Eqs.~\eqref{eqhelequal} and \eqref{eqChernoffsimple}, we are in a position to evaluate the quantities $R(M)$ and $\delta(M,N)$. In Fig.~\ref{fig:eg2_exponent} a), we plot $1-R(M)$ as a function of $M$. Although it is not possible to say definitively that $R(M)$ does not plateau to 1 for any finite $M$, it certainly appears that way from this figure. In Fig.~\mbox{\ref{fig:eg2_exponent} b)}, we plot $\delta(M+2,M)$ as a function of $M$. In this case it is evident that this function is continuously decreasing. However, again an analytic proof of this statement is not presented here.

Interestingly it appears that the purity of the unknown state strongly influences how rapidly the Chernoff bound can be approached. For more impure states closer to the maximally mixed state (and hence states which are harder to distinguish), the Chernoff bound is approached more rapidly than for more pure states. 

As this example is essentially classical (the two states are simultaneously diagonalisable), this appears to suggest that the inability to exactly saturate the Chernoff bound is not a quantum property.  Indeed similar concepts have been studied in classical information theory, where second order corrections to the asymptotic limits are known~\cite{hayashi2009information}.
\begin{figure*}[t]
\includegraphics[width=0.9\textwidth]{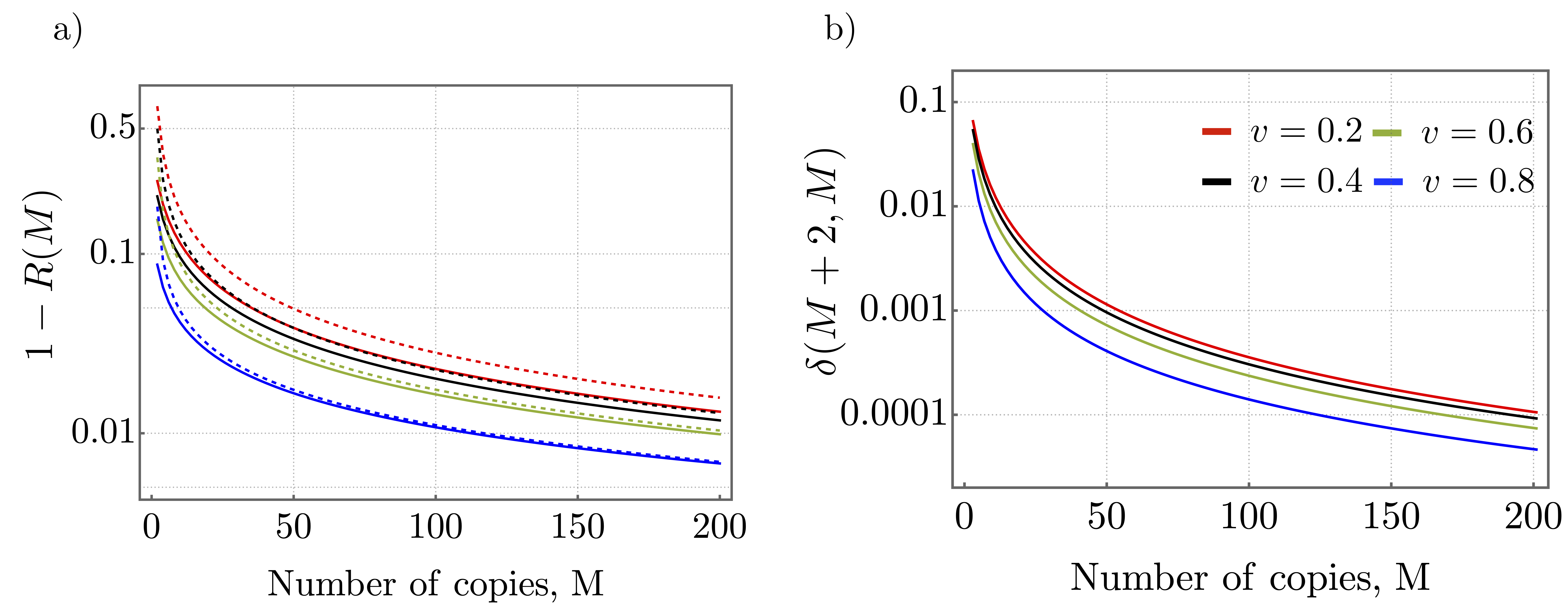}
\caption{\textbf{Scaling of error probabilities for a simple qubit example.} a) We plot 1 minus the ratio of the normalised $M$-copy error probability to the error probability attainable in the asymptotic limit (Eq.~\eqref{eq:ratioPe}). The solid lines correspond to even $M$ and the dashed lines to odd $M$. b) We plot the difference between the $M$ and $M+2$-copy ratios from a). The trend is shown only for even $M$, but the trend for odd $M$ is similar.}
\label{fig:eg2_exponent}
\end{figure*}

\subsection{Finite sample comparison of bounds in example 2}
Using Eqs.~\eqref{eqHelstromeg2} and ~\eqref{eqChernoffeg2}, we can evaluate the normalised ratio of the Helstrom bound to the Chernoff bound (Eq.~\eqref{eq:ratioPe}) as
\begin{equation}
\label{eq:ratiopure}
R(M)=\frac{\bigg(1-\sqrt{1-\abs{\bra{\psi_+}\ket{\psi_-}}^{2M}}\bigg)^{1/M}}{\text{cos}(\alpha)^{2}}\;.
\end{equation}
Below we prove that this ratio never reaches the value 1 for any finite $M$, hence the Chernoff bound cannot be exactly saturated. This is consistent with existing results for asymmetric quantum state discrimination~\cite{pereira2023analytical,audenaert2012quantum,rouze2017finite,li2014second}. Of course it is still possible that the Chernoff bound can be approximately saturated, as is shown in Fig.~\ref{fig:purequb}. In this figure, we plot 1 minus the normalised ratio of the $M$-copy error probability to the asymptotic error probability. This plot shows that the ratio is an increasing function of $M$, however the ratio never reaches exactly 1. Therefore, the asymptotic limit cannot be reached with any finite-copy measurement. This result is stated as the following theorem
\begin{theorem}
\label{th1}
For any two pure qubit states, the Chernoff bound cannot be saturated exactly for any finite $M$.
\end{theorem}
\begin{proof}

Starting from Eq.~\eqref{eq:ratiopure} we wish to prove that the Chernoff bound cannot be saturated with collective measurements on any finite number of copies. Eq.~\eqref{eq:ratiopure} gives the normalised ratio of the $M$-copy Helstrom bound to the probability of error attained from the $M$-copy Chernoff bound. 
\begin{equation}
\label{eq:ratiopureApen}
R(M)=\frac{\bigg(1-\sqrt{1-\abs{\bra{\psi_+}\ket{\psi_-}}^{2M}}\bigg)^{1/M}}{\text{cos}(\alpha)^{2}}\;.
\end{equation}
We are interested in proving that $1-R(M)>0$ for all finite $M$. We shall prove $1^M>R(M)^M$ for all finite $M$, which is equivalent to $1>R(M)$.\begin{equation}
\begin{split}
R(M)^M=\frac{\bigg(1-\sqrt{1-\text{cos}(\alpha)^{2M}}\bigg)}{\text{cos}(\alpha)^{2M}}\;.
\end{split}
\end{equation}
Denote $\gamma=\text{cos}(\alpha)^{2M}$,
\begin{equation}
\begin{split}
R(M)^M=\frac{\bigg(1-\sqrt{1-\gamma}\bigg)}{\gamma}\;.
\end{split}
\end{equation}
For $0<\gamma<1$, $\sqrt{1-\gamma}>1-\gamma$.
\begin{equation}
\begin{split}
R(M)^M=\frac{\bigg(1-\sqrt{1-\gamma}\bigg)}{\gamma}\\
<\frac{\bigg(1-(1-\gamma)\bigg)}{\gamma}=1\;.
\end{split}
\end{equation}

\end{proof}

\begin{figure}[t]
\includegraphics[width=0.5\textwidth]{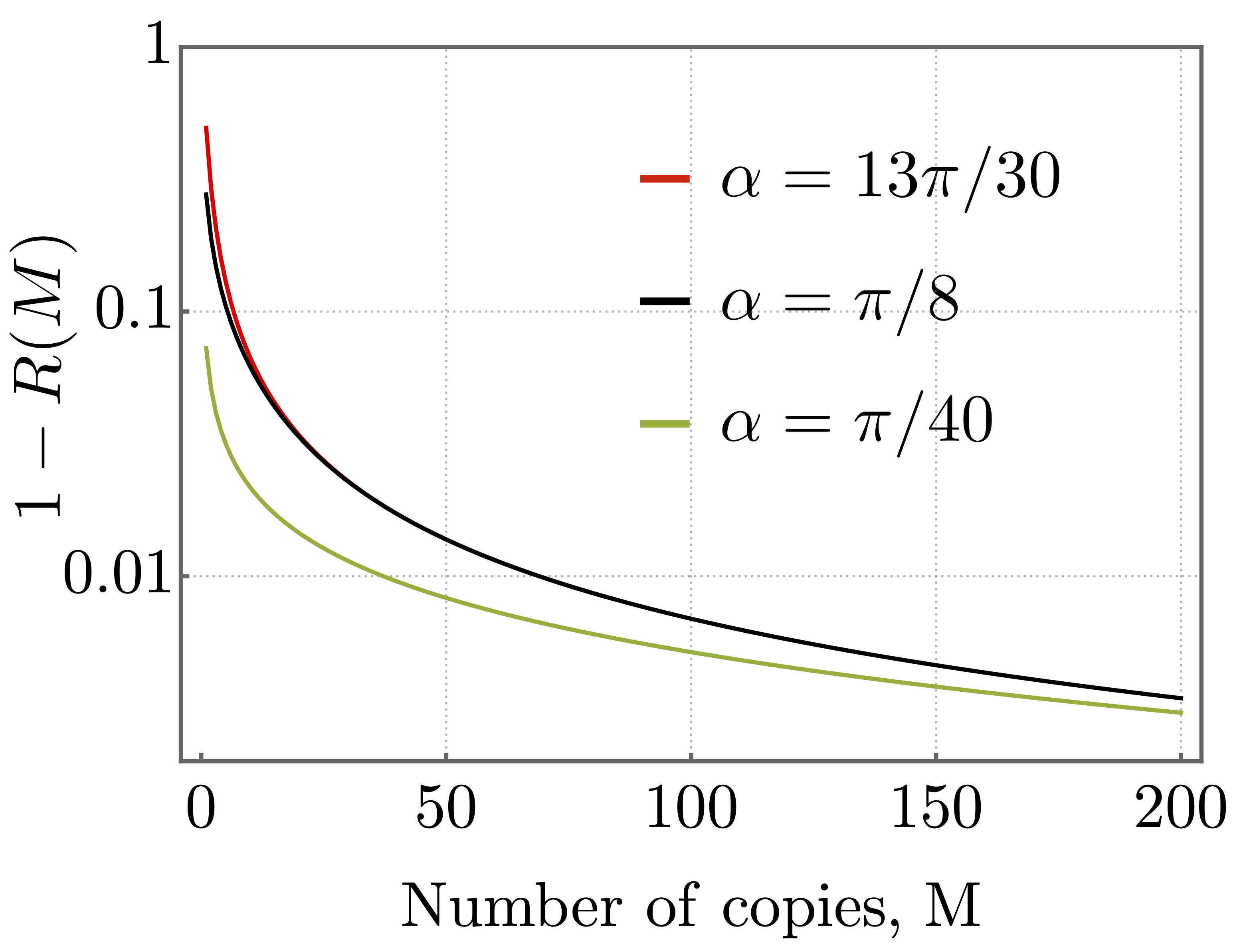}
\caption{\textbf{Scaling of the error probability ratio with increasing number of copies of pure qubit states.} 1 minus the normalised ratio of the error probability given by the Helstrom bound to that of the quantum Chernoff bound. This quantity approaches 0, however, as we prove in the main text, does not reach 0 for any finite value of $M$. }
\label{fig:purequb}
\end{figure}

\subsubsection{Numerical simulations for mixed states}
For mixed states we see similar trends, shown in Fig.~\ref{fig:mixedqub}. In the mixed case it is harder to obtain analytic results. However, as shown in Fig.~\ref{fig:mixedqub}, the numerical results demonstrate a similar trend to previously, namely that the Chernoff bound is approached more rapidly for less pure states. These simulations also show that the normalised ratio of error probabilities is an increasing function of $M$, however whether this ratio can exactly equal 1 for any finite $M$ remains an open question. 

\begin{figure}[t]
\includegraphics[width=0.48\textwidth]{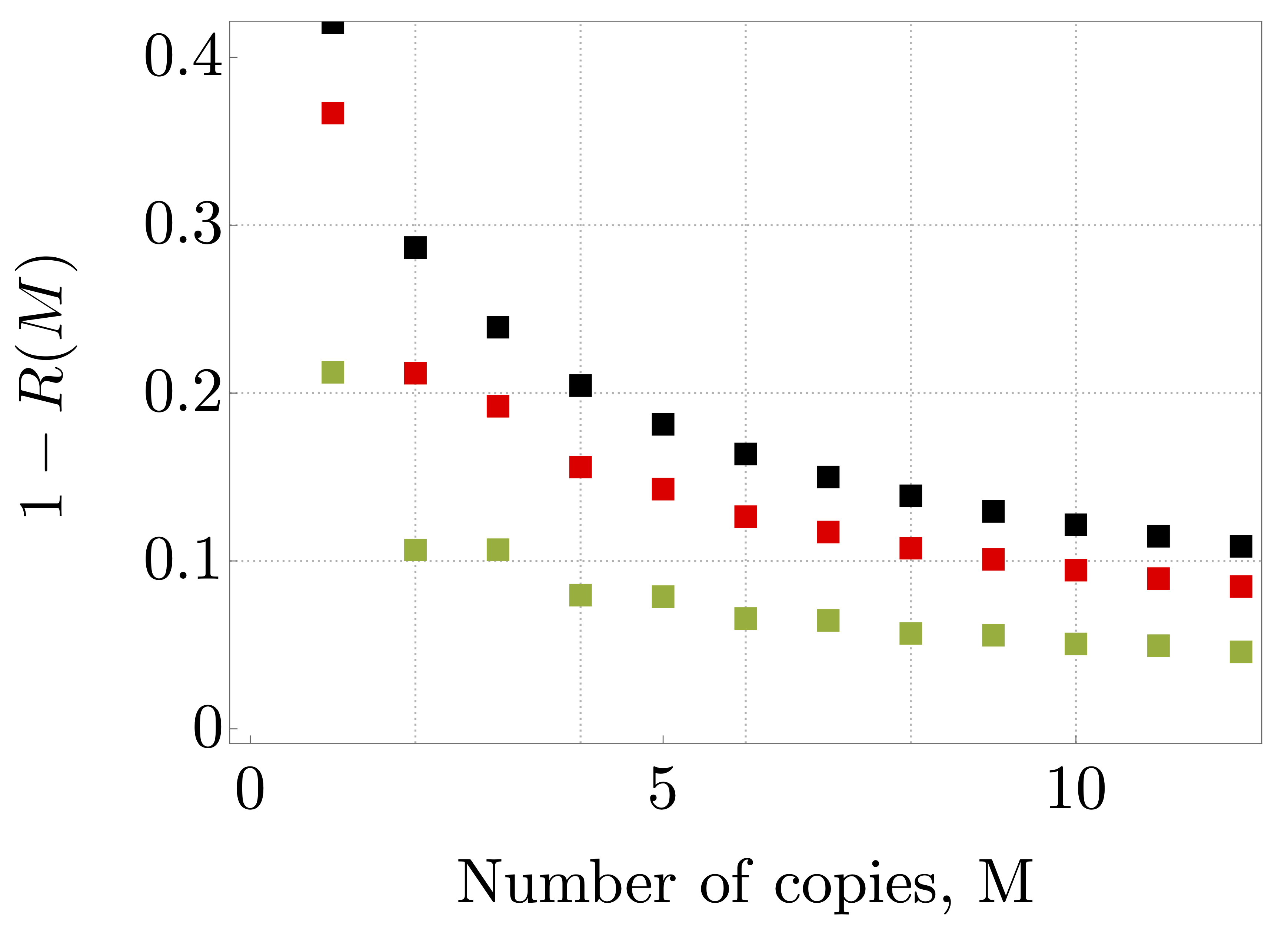}
\caption{\textbf{Comparing the error probability for mixed states.} We plot 1 minus the normalised ratio of the $M$-copy Helstrom bound error probability and the asymptotic error probability. Black, red and green markers correspond to $v=0.05$, $v=0.25$ and $v=0.6$ respectively. All points use $\alpha=\pi/5$.}
\label{fig:mixedqub}
\end{figure}

\bibliography{state_disc_bib}
\bibliographystyle{naturemag}

\end{document}